\documentclass{article}
\usepackage{graphicx} 
\usepackage{amsmath,amssymb,amsthm}
\usepackage{fullpage}
\usepackage[hcentering]{geometry}
\usepackage{caption} \usepackage{subcaption}
\usepackage{graphicx}
\usepackage{float}
\usepackage{placeins}
\usepackage{paralist}
\usepackage[utf8]{inputenc}   
\usepackage[T1]{fontenc} 
\newif\iffinal
\finaltrue
\newif\ifarxiv
\arxivfalse

\usepackage{paralist}
\usepackage{url}

\usepackage{aliascnt}
\def\NewTheorem#1#2{%
  \newaliascnt{#1}{theorem}
  \newtheorem{#1}[#1]{#2}
  \aliascntresetthe{#1}
  \expandafter\def\csname #1autorefname\endcsname{#2}
}
 \newtheorem{theorem}{Theorem}[section]
\NewTheorem{lemma}{Lemma}
\NewTheorem{corollary}{Corollary}
\NewTheorem{conjecture}{Conjecture}
\NewTheorem{observation}{Observation}
\NewTheorem{definition}{Definition}
\NewTheorem{assumption}{Assumption}

\NewTheorem{proposition}{Proposition}
\NewTheorem{remark}{Remark}
\NewTheorem{claim}{Claim}
\NewTheorem{fact}{Fact}

\newcommand{\hide}[1]{}

\DeclareMathOperator{\lmax}{\ell_{max}}

\usepackage{color}
\usepackage[usenames,dvipsname s,svgnames,table]{xcolor}

\definecolor{darkred}{rgb}{0.5,0,0}
\definecolor{lightblue}{rgb}{0,0.4,0.8}
\definecolor{darkgreen}{rgb}{0,0.5,0}

\usepackage[colorlinks=true, linkcolor=blue, urlcolor=blue, citecolor=darkgreen]{hyperref}

\usepackage[utf8]{inputenc}

\begin{document}

\title{Using Single-Neuron Representations for Hierarchical Concepts as Abstractions of Multi-Neuron Representations}
\author{Nancy Lynch\footnote{Massachusetts Institute of Technology.  This work was supported by NSF under grants CCR-2139936 and CCR-2003830.}}
\date{March 21, 2025}
\maketitle

\begin{abstract}
Brain networks exhibit many complications, such as noise, neuron failures, and partial synaptic connectivity.
These can make it difficult to model and analyze the behavior of the networks.  
This paper describes one way to address this difficulty, namely, breaking down the models and analysis using \emph{levels of abstraction}. 
We describe the approach for a particular problem, \emph{recognition of hierarchically-structured concepts}, as considered, for example, in~\cite{LM21,DBLP:conf/sirocco/LynchM23,lynch2024multineuron}.

Realistic models for representing hierarchical concepts use multiple neurons to represent each concept, in order to tolerate variations such as failures of limited numbers of neurons;
see, for example, early work by Valiant~\cite{Valiant}, work on the assembly calculus~\cite{DBLP:conf/innovations/Legenstein0PV18,PapadimitriouVempala}, and our own recent work~\cite{lynch2024multineuron}.
These models are intended to capture some behaviors present in actual brains; however, analysis of mechanisms based on multi-neuron representations can be somewhat complicated.
On the other hand, simpler mechanisms based on single-neuron representations can be easier to understand and analyze~\cite{LM21,DBLP:conf/sirocco/LynchM23}, but are less realistic.

In this paper, we show that these two types of models are compatible, and in fact, networks with single-neuron representations can be regarded as formal, mathematical abstractions of networks with multi-neuron representations.
%
We do this by revisiting networks with multi-neuron representations like those in~\cite{lynch2024multineuron}, and relating them formally to networks with single-neuron representations like those in~\cite{LM21}.
As in~\cite{LM21,lynch2024multineuron}, we consider the problem of recognizing hierarchical concepts based on partial information. 

Specifically, we consider two networks, $\mathcal H$ and $\mathcal L$, with multi-neuron representations, one with high connectivity and one with low connectivity.
We define two abstract networks, $\mathcal{A}_1$ and $\mathcal{A}_2$, with single-neuron representations, and prove that they recognize concepts correctly.
Then we prove correctness of $\mathcal H$ and $\mathcal L$ by relating them formally to $\mathcal{A}_1$ and $\mathcal{A}_2$.
In this way, we decompose the analysis of each multi-neuron network into two parts:  analysis of abstract, single-neuron networks, and proofs of formal relationships between the multi-neuron network and single-neuron networks.

These examples illustrate what we consider to be a promising, tractable approach to analyzing other complex brain mechanisms.
\end{abstract}

\section{Introduction}

We continue our work in~\cite{LM21,DBLP:conf/sirocco/LynchM23,lynch2024multineuron} on representing hierarchically-structured concepts in layered Spiking Neural Networks, in such a way that concepts can be recognized based on only partial input information.
Our first paper~\cite{LM21} began the study by considering simple representations of concepts in feed-forward networks with total connectivity between layers.
There, each concept is represented by a single reliable neuron, at a layer of the network corresponding to the level of the concept.  The focus in that paper is on learning the representations, in both noise-free and noisy settings.
Next~\cite{DBLP:conf/sirocco/LynchM23}, we extended the recognition work in~\cite{LM21} to concepts that allow some exceptions to strict hierarchical structure and networks that include some feedback edges.

Most recently~\cite{lynch2024multineuron}, we extended the recognition work of~\cite{LM21} in order to tolerate limited numbers of neuron failures and partial network connectivity, in addition to partial input information.
Now, in order to tolerate failures and disconnections, our networks use multiple neurons to represent each concept.  
This is reminiscent of work on assembly calculus~\cite{DBLP:conf/innovations/Legenstein0PV18,PapadimitriouVempala,PVMM20}, and is realistic for concept representations in actual brains.

The proofs of our networks in~\cite{lynch2024multineuron} are somewhat complicated, which led us to look for ways of simplifying them by decomposing them into smaller pieces.
One approach that has been used successfully in the distributed algorithms and formal methods communities is to describe and prove properties of systems using \emph{levels of abstraction}.
We wondered whether we could view the detailed networks of~\cite{lynch2024multineuron}, which use multiple neurons to represent each concept in order to tolerate neuron failures and edge disconnections, as "implementations" of more abstract networks based on a single reliable neuron for each concept.  
This paper demonstrates that this can work.

Our first observation was that the models for failures and disconnections in~\cite{lynch2024multineuron} are probabilistic, based on random initial failures of neurons and random connections between layers.
We considered introducing a small probability of failure into the abstract networks, and using a probabilistic mapping between detailed and abstract networks, similar to mappings in~\cite{DBLP:journals/njc/SegalaL95}.  
But this approach does not seem so natural, since it involves preserving exact probabilities, which seems like too stringent a requirement for our purposes.
For example,  failure probabilities are normally not intended to describe exact probabilities, but as overestimates.

So we took a different approach:  extract the probabilities from each detailed network, leaving us with a network that behaves deterministically, for each particular failure pattern, connectivity pattern, and input.
The failure pattern and connectivity pattern are required to satisfy certain constraints, and the probabilistic network should "behave like" the deterministic network with high probability.
Then establish a formal correspondence between the detailed deterministic network and an abstract deterministic network (with the same input). Since this involves a mapping between two deterministic networks, the mapping could be a simpler type of (non-probabilistic) mapping, similar to those used in~\cite{DBLP:conf/podc/LynchT87}. 

In our case, the detailed networks include the constraint that at least $m (1-\epsilon)$ of the $reps$ of each concept "survive", that is, do not fail.
Here, $m$ is the number of $reps$ for each concept, 
and $\epsilon$ is a small \emph{recognition approximation} parameter.
The value of $\epsilon$ can be computed as $1 - p (1-\zeta)$, 
where $p = 1-q$, $q$ is the neuron failure probability for the probabilistic version of the detailed network, and $\zeta$ is a concentration parameter for a Chernoff bound.
For the case with limited connectivity, the detailed network also includes the constraint that for every $v \in reps(c)$, and each child $c'$ of $c$, 
at least $a m (1-\epsilon)$ surviving $reps$ of $c'$ are connected to $v$. 
Here, $a$ is a coefficient used to capture the notion of partial connectivity.

So, we wanted to relate our deterministic detailed networks, with failures and disconnections, to a deterministic abstract network $\mathcal A$ with no failures or disconnections.  
We wanted to prove separately that the abstract network $\mathcal A$ is correct and that the behavior of the detailed networks corresponds closely to that of $\mathcal A$; then together this should yield that the detailed networks behave correctly.
That would further imply that, with high probability, the probabilistic versions of the detailed networks behave correctly.

But that is not quite how it worked out.
As it happens, correctness for the recognition problem has two aspects:  with enough partial information, concepts \emph{should} be recognized, and with too little partial information, the concepts \emph{should not} be recognized.  
That is, we have a firing requirement and a non-firing requirement.
There is also a "middle ground":  with intermediate amounts of information, either outcome is permissible.
To cope with this middle ground, we found it convenient to consider two abstract networks rather than just one:  a network $\mathcal{A}_1$ that guarantees the firing requirement and a network $\mathcal{A}_2$ that guarantees the non-firing requirement.

We consider two detailed networks, $\mathcal{H}$ with high connectivity and $\mathcal{L}$ with low connectivity, both based on detailed probabilistic networks from~\cite{lynch2024multineuron}.
We show that each of these corresponds to both $\mathcal{A}_1$ and $\mathcal{A}_2$, which implies that both networks $\mathcal{H}$ and $\mathcal{L}$ perform correct recognition.
In order to establish these correspondences, we define formal notions of \emph{implementation} to map detailed networks to abstract networks while preserving correct behavior.  


\paragraph{Outline of the paper:}
In Section~\ref{sec: concepts}, we define our model for hierarchies of concepts, including the notion of when certain partial information is sufficient to "support" recognition.
In Section~\ref{sec: networks}, we define our general Spiking Neural Network model, which we use to describe all of our networks, both abstract and detailed.
Section~\ref{sec: prob-recog} contains our definitions for the recognition problem; we have two separate definitions, one for networks in which concepts have single-neuron representations and one for networks with multi-neuron representations.
Section~\ref{sec: implementation} contains our definitions of implementation relationships between detailed networks and abstract networks.
Section~\ref{sec: abstract-algorithms} contains definitions of our abstract networks $\mathcal{A}_1$ and $\mathcal{A}_2$ and proofs that they work together to solve the recognition problem for hierarchical concepts.
Section~\ref{sec: detailed-algorithm-high} contains definitions for our first detailed network, $\mathcal H$, with high connectivity, and proves its correctness by showing that it implements both $\mathcal{A}_1$ and $\mathcal{A}_2$.
Section~\ref{sec: detailed-algorithm-low} contains definitions for our second detailed network, $\mathcal L$, with low connectivity, and proves its correctness, again by showing that it implements both $\mathcal{A}_1$ and $\mathcal{A}_2$.
Section~\ref{sec: conclusions} contains our conclusions.

\section{Concept Model}
\label{sec: concepts}

We define concept hierarchies as in~\cite{LM21}.  
In general, we think of a concept hierarchy as containing all of the concepts that have been learned by an organism over its lifetime.

\subsection{Preliminaries}

In referring to concept hierarchies, we use the following parameters:
\begin{itemize}
\item $\ell_{max}$, a positive integer, representing the maximum level number for the concepts that we consider.
\item $n$, a positive integer, representing the total number of lowest-level concepts that we consider.
\item $k$, a positive integer, representing the number of top-level concepts in any concept hierarchy, and also the number of child concepts for each concept whose level is $\geq 1$.
\item $r_1, r_2$, reals in $[0,1]$ with $r_1 \leq r_2$; these represent thresholds for noisy recognition.
\end{itemize}

We assume a universal set $D$ of \emph{concepts}, partitioned into disjoint sets $D_{\ell}$,  $0 \leq \ell \leq \lmax$.
We refer to any particular concept $c \in D_{\ell}$ as a \emph{level} $\ell$ \emph{concept}, and write $level(c) = \ell$.
Here, $D_0$ represents the most basic concepts and $D_{\lmax}$ the highest-level concepts.
We assume that $|D_0| = n$.

\subsection{Concept hierarchies}

A \emph{concept hierarchy} $\mathcal C$ consists of a subset $C$ of $D$, together with a $children$ function.  For each $\ell$, $0 \leq \ell \leq \lmax$, we define $C_{\ell}$ to be $C \cap D_{\ell}$, that is, the set of level $\ell$ concepts in $\mathcal C$.
For each concept $c \in C_{\ell}$, $1 \leq \ell \leq \ell_{max}$, we designate a nonempty set $children(c) \subseteq C_{\ell-1}$.
We call each $c' \in children(c)$ a \emph{child} of $c$.
We assume the following properties.

\begin{enumerate}
\item
$|C_{\lmax}| = k$; that is, the number of top-level concepts is exactly $k$.
\item
For any $c \in C_{\ell}$, where $1 \leq \ell \leq \lmax$, we have that $|children(c)| = k$; that is, the degree of any internal node in the concept hierarchy is exactly $k$.
\item
For any two distinct concepts $c$ and $c'$ in $C_{\ell}$, where $1 \leq \ell \leq \lmax$, we have that $children(c) \cap children(c') = \emptyset$; that is, the sets of children of different concepts at the same level are disjoint.\footnote{Thus, we allow no overlap between the sets of children of different concepts. We study overlap in~\cite{DBLP:conf/sirocco/LynchM23}.}
\end{enumerate}
Thus, a concept hierarchy ${\mathcal C}$ is a forest with $k$ roots and height $\lmax$.  Of course, this is a drastic simplification of any real concept hierarchy, but the uniform structure makes networks easier to analyze.

We extend the $children$ notation recursively by defining a concept $c'$ to be a $descendant$ of a concept $c$ if either $c' = c$, or $c'$ is a child of a descendant of $c$.
We write $descendants(c)$ for the set of descendants of $c$.
Let $leaves(c) = descendants(c) \cap C_0$, that is, all the level 0 descendants of $c$.

\subsection{Support}
\label{sec:support}

Now we define which sets of level $0$ concepts provide enough partial information to "support" recognition of higher-level concepts.  

We fix a particular concept hierarchy $\mathcal C$, with its concept set $C$ partitioned into $C_0,\ldots,C_{\lmax}$.
For any given subset $B$ of the universal set $D_0$ of level $0$ concepts, and any real number $r \in [0,1]$, we define the set $supp_r(B)$ of concepts in $C$.
This is intended to represent the set of concepts $c \in C$ at all levels that have enough of their leaves present in $B$ to support recognition of $c$. 
The notion of "enough" here is defined recursively, based on having an $r$-fraction of children supported for every concept at every level.


\begin{definition}[\textbf{Supported}]
\label{def:support}
Given $B \subseteq D_0$, define the following sets of concepts at all levels, recursively:
\begin{enumerate}
\item
$B(0) = B \cap C_0$. 
\item
For $1 \leq \ell \leq \ell_{max}$, $B(\ell)$ is the set of all concepts $c \in C_{\ell}$ such that $|children(c) \cap B(\ell - 1)|  \ \geq \ r k$. 
\end{enumerate}
Define $supp_r(B)$ to be $\bigcup_{0 \leq \ell \leq \lmax} B(\ell)$.  
We say that each concept in $supp_r(B)$ is $r$-$supported$ by $B$.
\end{definition}

\section{Spiking Neural Network Model}
\label{sec: networks}

We consider feed-forward networks in which all edges point from neurons in one layer to neurons in the next-higher-numbered layer.

We consider neuron failures during the recognition process.
For our failure model, we consider initial stopping failures: if a neuron fails, it never performs any activity, that is, it never updates its state and never fires. 

We do not consider learning in this paper, only recognition.  So we omit all aspects of our previous network model~\cite{LM21} that involve learning.

\subsection{Preliminaries}

Throughout the paper, in referring to our networks, we use the following parameters:
\begin{itemize}
     \item
     $\ell'_{max}$, a positive integer, representing the maximum number of a layer in the network.
     \item 
     $width$, a mapping from the set of layer numbers $\{1,\ldots,\ell'_{max}\}$ to nonnegative integers, giving the number of neurons in each layer.
     \item
     $\tau$, a nonnegative real, representing the firing threshold for neurons.
\end{itemize}

\subsection{Network structure}
\label{sec: network-structure}

Our networks are directed graphs consisting of neurons arranged in layers, with forward edges directed from each layer to the next-higher layer.  
Specifically, a network $\mathcal{N}$ consists of a set $N$ of neurons, partitioned into disjoint sets $N_{\ell}, 0 \leq \ell \leq \ell'_{max}$, which we call \emph{layers}.
We refer to any particular neuron $u \in N_{\ell}$ as a \emph{layer} $\ell$ \emph{neuron}, and write $layer(u) = \ell$.
We refer to the layer $0$ neurons as \emph{input neurons}.

We assume total connectivity between successive layers, that is, each neuron $u \in N_{\ell}$, $0 \leq \ell \leq \ell'_{max} - 1$, has an outgoing edge to each neuron $v \in N_{\ell+1}$.  In this paper, these are the only edges.

We assume that the state of each neuron consists of several state components:
\begin{itemize}
    \item \emph{firing}, with values in  $\{0,1\}$; this indicates whether or not the neuron is currently firing, where $1$ indicates that it is firing and $0$ indicates that it is not firing.
    \item \emph{failed}, with values in  $\{0,1\}$; this indicates whether or not the neuron has failed,where $1$ indicates that it has failed and $0$ indicates that it has not failed.  In this paper we consider initial stopping failures only.
\end{itemize}
We denote the \emph{firing} component of neuron $u$ at integer time $t$ by $firing^u(t)$ and the \emph{failed} component of neuron $u$ at time $t$ by $failed^u(t)$.  In this paper, failures occur only at the start, thus, each $failed^u$ component is constant over time.  

For our abstract networks, we will not consider failures; for this case, to simplify matters, we omit
the $failed$ component.

Each non-input neuron $u \in N_{\ell}$, $1 \leq \ell \leq \ell'_{max}$, has an additional state component:
\begin{itemize}
\item
\emph{weight}, a real-valued vector in $\{0,1\}^n$ representing the weights of all incoming edges.
\end{itemize}
We denote this component of non-input neuron $u$ at time $t$ by $weight^u(t)$.

\subsection{Network operation}
\label{sec: network-operation}

The network operation is determined by the behavior of the individual neurons.
We distinguish between input neurons and non-input neurons.

If $u$ is an \emph{input neuron}, then it has only two state components, $failed$ and $firing$.
Since $u$ is an input neuron, we assume that the values of both $failed^u$ and $firing^u$ at all times $t$ are 
controlled by the network's environment and not by the network itself; that is, the values of $failed^u(t)$ and $firing^u(t)$ are set by some external force, which we do not model explicitly.
Since we assume initial stopping failures, the value of $failed^u$ is the same at every time $t$, that is, $failed^u(t) = failed^u(0)$ for every $t \geq 0$.
We assume that if an input neuron fails, it never fires, that is, $failed^u(0) = 1$ implies that $firing^u(t) = 0$ for every $t \geq 0$.

If $u$ is a \emph{non-input neuron}, then it has three state components, $failed$, $firing$, and $weight$.
The values of $failed^u$ at all times $t$ are set by an external force, as for input neurons.  Again, since we assume initial stopping failures, the value of $failed^u$ is the same at every time $t$, that is, $failed^u(t) = failed^u(0)$ for every $t \geq 0$.
A non-input neuron $u$ that fails never fires, that is, $failed^u(0) = 1$ implies that $firing^u(t) = 0$ for every $t \geq 0$.

For a non-input neuron $u$ that does not fail, the value of $firing^u(0)$ is determined by the initial network setting, whereas the value of $firing^u(t)$, $t \geq 1$, is determined by $u$'s incoming \emph{potential} and its \emph{activation function}.
To define the potential, let $x^u(t)$ denote the vector of $firing$ values of $u$'s incoming neighbor neurons at time $t$.
These are all the nodes in the layer numbered $layer(u) - 1$.
Then the potential for time $t$, $pot^u(t)$, is given by the dot product of the $weight$ vector and incoming firing pattern at neuron $u$ at time $t-1$, that is, 
\[
pot^u(t) = weight^u(t-1) \cdot x^u(t-1) = \sum_j weight^u_j(t-1) x^u_j(t-1)j,
\]
where $j$ ranges over the set of incoming neighbors.
The activation function, which specifies whether or not neuron $u$ fires at time $t$, is then defined by:
\[ 
firing^u(t) =  
\begin{cases}
1 & \text{if } pot^u(t) \geq \tau, \\
0 & \text{if } \text{otherwise},
\end{cases}
\]
where $\tau$ is the assumed firing threshold.

For a non-input neuron $u$, the value of $weight^u(0)$ is determined by the initial network setting.
In this paper, we assume that the $weight$ vector remains unchanged: $weight^u(t) = weight^u(0)$ for every $t \geq 0$.

During execution, the network proceeds through a sequence of \emph{configurations}, where each configuration specifies a state for every neuron in the network, that is, values for all the state components of every neuron.

\section{Two Recognition Problems}
\label{sec: prob-recog}

In this section, we define our recognition problem formally.
We define two versions of the problem, one for "abstract" networks with single-neuron representations and one for "detailed" networks with multi-neuron representations.
We need two different definitions because we assume different conventions for input and output for these two cases.
In Section~\ref{sec: implementation} we describe how these two definitions are related.

The input to the networks involves "presenting" a set $B$ of level $0$ concepts, according to a presentation convention.
This convention is different for abstract networks with single-neuron representations and detailed networks with multi-neuron representations.
The output requirements are expressed in terms of firing guarantees for certain sets of neurons.
We do not consider probabilities here, but express the behavior just in terms of sets of firing neurons. 
In realistic versions of the multi-neuron networks, these sets of neurons could be said to fire only with high probability, not with certainty.
However, we leave these probabilities to be analyzed elsewhere, for example,~\cite{lynch2024multineuron}.

We consider a particular concept hierarchy $\mathcal C$,
with concept hierarchy notation as defined in Section~\ref{sec: concepts}.\footnote{This is slightly different from what we did in~\cite{lynch2024multineuron}, where we restricted attention to a single concept $c$ and all its descendants, i.e., a tree rather than a $k$-root forest.  Here we use the full concept hierarchy, with $k$ roots.  The main difference in the results is that we have more concepts in the concept hierarchy, which would require the probability bounds from~\cite{lynch2024multineuron} to be adjusted accordingly.  Since we are not focusing on probabilities in this paper, this will not affect our results.}
For our networks, we use notation as defined in Section~\ref{sec: networks}.

\subsection{Recognition for a single-neuron representation}
\label{sec: recog-single}

In this subsection, we define what it means for an "abstract" network $\mathcal A$ that uses a single-neuron representation to recognize concept hierarchy $\mathcal C$.

The definition assumes that every level $0$ concept $c$ has a single designated representing layer $0$ neuron $rep(c)$; this will be used to provide $c$ as an input to the network.
For any $B \subseteq D_0$, we define $reps(B) = \{ rep(b) \ | \ b \in B \}$.
That is, $reps(B)$ is the set of all $reps$ of concepts in $B$.
Furthermore, we assume that every concept $c \in C$ with $level(c) \geq 1$ has a representing neuron, $rep(c)$, in some layer $\geq 1$ of the network.
All of the $rep(c)$ neurons are distinct.

Note that these conventions require that $width(0) \geq |D_0|$, that is, the number of neurons in layer $0$, is at least  $|D_0|$.  Similarly, there must be sufficiently many neurons in higher layers to accommodate all of the $reps$.

Our recognition problem definition relies on the following definition of how a particular set $B$ of level $0$ concepts is ``presented'' to the network.  This involves firing exactly the input neurons that represent these level $0$ concepts.

\begin{definition}[\textbf{Presented}]
\label{def: presented1}
If $B \subseteq D_0$ and $t \geq 0$, then we say that $B$ is \emph{presented at time} $t$ (in some particular network execution) exactly if the following holds.
For every layer $0$ neuron $u$:
\begin{enumerate}
    \item  If $u \in reps(B)$, then $u$ fires at time $t$.
    \item  If $u \notin reps(B)$, then $u$ does not fire at time $t$.
\end{enumerate}
\end{definition}

Now we can define what it means for network $\mathcal A$ to recognize concept hierarchy $\mathcal C$.  The definition says that, for each concept $c$ that is $r_2$-supported by $B$, $rep(c)$ must fire.
On the other hand, if $c$ is not $r_1$-supported by $B$, then $rep(c)$ must not fire.
To simplify things a bit, we assume here that $B \subseteq C_0$, that is, $B$ is a subset of the actual level $0$ concepts in the concept hierarchy $\mathcal C$.

\begin{definition}[\textbf{Recognition problem for single-neuron representations}]
\label{def: recog-ff-single}
Network $\mathcal A$ $(r_1,r_2)$-\emph{recognizes} $\mathcal C$ provided that 
the following holds.
Assume that $B \subseteq C_0$ is presented at time $0$.  
Then:
\begin{enumerate}
\item
\emph{When $rep(c)$ should fire:}
If $c \in supp_{r_2}(B)$, then neuron $rep(c)$ fires at time $layer(rep(c))$.
\item
\emph{When $rep(c)$ should not fire:}
If $c \notin supp_{r_1}(B)$, then neuron $rep(c)$ does not fire at time $layer(rep(c))$.
\end{enumerate}
\end{definition}

For use later in the paper, we find it convenient to split this two-part definition into two parts: 

\begin{definition}[\textbf{Firing guarantee for single-neuron representations}]
\label{def: firing-guarantee}
Network $\mathcal A$ \emph{guarantees} $r_2$-\emph{firing} for $\mathcal C$ provided that the following holds.
If $B \subseteq C_0$ is presented at time $0$ and if $c \in supp_{r_2}(B)$, then neuron $rep(c)$ fires at time $layer(rep(c))$.
\end{definition}

\begin{definition}[\textbf{Non-firing guarantee for single-neuron representations}]
\label{def: non-firing-guarantee}
Network $\mathcal A$ \emph{guarantees} $r_1$-\emph{non-firing} for $\mathcal C$ provided that
the following holds.
If $B \subseteq C_0$ is presented at time $0$ and if $c \notin supp_{r_1}(B)$, then neuron $rep(c)$ does not fire at time $layer(rep(c))$.
\end{definition}

\subsection{Recognition for a multi-neuron representation}
\label{sec: recog-multi}

Now we define what it means for a "detailed" network $\mathcal D$ that uses a multi-neuron representation to recognize concept hierarchy $\mathcal C$.
The definition assumes that every level $0$ concept $c$ has exactly $m$ layer $0$ neurons $reps(c)$.
For any $B \subseteq D_0$, we define $reps(B) = \bigcup_{b \in B} reps(b)$.
That is, $reps(B)$ is the set of all $reps$ of concepts in $B$.
Furthermore, we assume that every concept $c \in C$ with $level(c) \geq 1$ has a set of $m$ representing neurons, $reps(c)$, in layers $\geq 1$ of the network.
All of these $reps$ sets are disjoint.  

Note that this requires that $width(0) \geq |D_0| \ m$, that is, the number of neurons in layer $0$ is at least  $|D_0| \ m$.  Also, there must be sufficiently many neurons in higher layers to accommodate all of the $reps$.

Our recognition problem for multi-neuron representations uses the following new parameter:
\begin{itemize}
\item  
$\epsilon \in [0,1]$; this is the \emph{recognition approximation} parameter, representing a fraction of $rep$ neurons that might not fire.
\end{itemize}

We assume that the set $F$ of failed neurons (including the input neurons) is determined at the start.
For every neuron $u \in F$, the $failed$ flag of $u$ is set to $1$ at time $0$ and remains $1$ thereafter, that is, $failed^u(t) = 1$ for every $t \geq 0$.
For every neuron $u \notin F$, the $failed$ flag of $u$  is set to $0$ at time $0$ and remains $0$ thereafter, that is, $failed^u(t) = 0$ for every $t \geq 0$.

Our recognition problem definition relies on the following definition of how a particular set $B$ of level $0$ concepts is presented to the network.  This involves firing exactly the input neurons that represent these level $0$ concepts and have not failed.

\begin{definition}[\textbf{Presented}]
\label{def: presented2}
If $B \subseteq D_0$ and $t \geq 0$, then we say that $B$ is \emph{presented at time} $t$ (in some particular network execution) exactly if the following holds.
For every layer $0$ neuron $u$:
\begin{enumerate}
    \item  If $u \in reps(B) - F$, then $u$ fires at time $t$.
    \item  If $u \notin reps(B)$ or $u \in F$, then $u$ does not fire at time $t$.
\end{enumerate}
\end{definition}

Now we can define what it means for network $\mathcal D$ to recognize concept hierarchy $\mathcal C$.
We require that, for each concept $c$ that is $r_2$-supported by $B$, then
at least $m (1 - \epsilon)$ of the neurons in $reps(c)$ must fire.
On the other hand, if $c$ is not $r_1$-supported by $B$, then none of the neurons in $reps(c)$ should fire.

\begin{definition}[\textbf{Recognition problem for multi-neuron representations}]
\label{def: recog-ff-multi}
Network $\mathcal D$ $(r_1,r_2)$-\emph{recognizes} $\mathcal C$ provided that
the following holds. 
Assume that $B \subseteq C_0$ is presented at time $0$.  
Then:
\begin{enumerate}
\item
\emph{When $reps(c)$ neurons should fire:}
If $c \in supp_{r_2}(B)$, then 
at least $m (1 - \epsilon)$ of the neurons $v \in reps(c)$ fire at time $layer(v)$.
\item
\emph{When $reps(c)$ neurons should not fire:}
If $c \notin supp_{r_1}(B)$, then no neuron $v \in reps(c)$ fires at time $layer(v)$.
\end{enumerate}
\end{definition}

Unlike in the single-neuron case, we will not need to split this definition into two parts. 
We find such a split to be useful in this paper for abstract networks, but not for detailed networks.  We will show that each of our detailed networks implements two different abstract networks, one for each part of the recognition requirement.

\section{Relating Solutions to the Two Recognition Problems}
\label{sec: implementation}

In this section, we define formal implementation relationships between a detailed network $\mathcal D$ and an abstract network $\mathcal A$.
We would like to use these to show that, if $\mathcal A$ solves the recognition problem for single-neuron representations, then $\mathcal D$ solves the recognition problem for multi-neuron representations.

Throughout the rest of this paper, we fix a particular concept hierarchy $\mathcal C$.
We assume that, in $\mathcal A$, every concept $c \in C$ has one representing neuron, $rep(c)$, as specified in Section~\ref{sec: recog-single}.
In $\mathcal D$, every concept $c$ has a size-$m$ set of representing neurons, $reps(c)$, as specified in Section~\ref{sec: recog-multi}.
$\mathcal A$ does not admit any neuron failures, whereas $\mathcal D$ allows failures, as defined by a set $F$ of failed neurons.

We will define these implementation relationships in terms of individual executions of the two networks.  To do this, we remove some complications, by assuming that the layers in the networks correspond to the levels in the hierarchy.  That is:
\begin{itemize}
\item $\ell'_{max} = \ell_{max}$; the number of layers in the network is the same as the number of levels in the concept hierarchy.
    \item  In $\mathcal A$, for every concept $c$, $layer(rep(c)) = level(c)$.
    \item  In $\mathcal D$, for every concept $c$ and every $v \in reps(c)$, $layer(v) = level(c)$.
\end{itemize}
Moreover, we assume that, in both networks $\mathcal A$ and $\mathcal D$, the input set $B \subseteq C_0$ is presented at time $0$, and no firing of layer $0$ neurons occurs at any other time. 

With these assumptions, network $\mathcal A$ behaves deterministically given a particular input set $B$, and network $\mathcal D$ behaves deterministically given a particular set $F$ of failed neurons and input set $B$.

\subsection{First attempt}

Here is our first attempt at an approach to proving the correctness condition in Definition~\ref{def: recog-ff-multi} using implementation relationships between abstract and detailed networks.
We assume an abstract network $\mathcal{A}$ and a detailed network $\mathcal{D}$. 
$\mathcal D$ includes a fixed set $F$ of failed neurons. 
Constraints on $F$ will be defined in Sections~\ref{sec: detailed-algorithm-high} and~\ref{sec: detailed-algorithm-low}, for our particular detailed networks $\mathcal H$ and $\mathcal L$.

\begin{definition}[$Implements$, $\leq_{impl}$]
\label{def: 2-way implementation}
Network $\mathcal D$ \emph{implements} network $\mathcal A$ provided that, for every $B \subseteq C_0$, for the two unique executions of these two networks on $B$, the following holds: 
For every concept $c \in C$,
\begin{enumerate}
\item 
If $rep(c)$ fires at time $level(c)$ in $\mathcal A $, then in $\mathcal D$, at least $m (1-\epsilon)$ of the neurons $v \in reps(c)$ fire, each such $v$ at time $level(c)$.
\item 
If $rep(c)$ does not fire at time $level(c)$ in $\mathcal A$, then in $\mathcal{D}$, no neuron $v \in reps(c)$ fires at time $level(c)$.
\end{enumerate}
We write $\mathcal{D} \leq_{impl} \mathcal{A}$ as shorthand for $\mathcal D$ \emph{implements} $\mathcal A$.
\end{definition}
Note that Part 1 of the definition is stated in terms of the approximation parameter $\epsilon$.

This implementation definition is enough to infer correctness of the detailed network $\mathcal D$ from correctness of the abstract network $\mathcal A$:

\begin{theorem}
\label{thm: general-mapping-theorem}
Let $r_1$ and $r_2$ be reals in $[0,1]$ with $r_1 \leq r_2$.
If $\mathcal A$ $(r_1,r_2)$-recognizes $\mathcal C$ according to Definition~\ref{def: recog-ff-single} and $\mathcal{D} \leq_{impl} \mathcal{A}$, then $\mathcal D$ $(r_1,r_2)$-recognizes $\mathcal C$ according to Definition~\ref{def: recog-ff-multi}.
\end{theorem}

\begin{proof}
Assume that $\mathcal A$ $(r_1,r_2)$-recognizes $\mathcal C$ and $\mathcal{D} \leq_{impl} \mathcal{A}$. 
To show that $\mathcal D$ $(r_1,r_2)$-recognizes $\mathcal C$, we   
show the two parts of Definition~\ref{def: recog-ff-multi} separately.
Fix some particular $B \subseteq C_0$.

First, suppose that $c \in supp_{r_2}(B)$.  Then by Part 1 of Definition~\ref{def: recog-ff-single}, $rep(c)$ fires at time $level(c)$ in $\mathcal A$.
Then Part 1 of Definition~\ref{def: 2-way implementation} implies that, in $\mathcal D$, at least $m (1-\epsilon)$ of the neurons $v \in reps(c)$ fire, each such $v$ at time $level(c)$. This shows Part 1 of Definition~\ref{def: recog-ff-multi}.

Second, suppose that $c \notin supp_{r_1}(B)$.
Then by Part 2 of Definition~\ref{def: recog-ff-single}, $rep(c)$ does not fire at time $level(c)$ in $\mathcal A$.
Then Part 2 of Definition~\ref{def: 2-way implementation} implies that, in $\mathcal D$, no neuron in $reps(c)$ fires at time $level(c)$.  This shows Part 2 of Definition~\ref{def: recog-ff-multi}.
\end{proof}

In Sections~\ref{sec: abstract-algorithms}-\ref{sec: detailed-algorithm-low}, we define particular abstract and detailed networks.  
We would like to apply Theorem~\ref{thm: general-mapping-theorem} to show that the detailed networks solve the multi-neuron recognition problem.
This would entail showing that the abstract networks solve the single-neuron recognition problem and that the detailed networks are related to the abstract networks using the implementation relation $\leq_{impl}$.

However, that is not quite how it worked out.
Instead of a single abstract network $\mathcal A$, we found it convenient to use two abstract networks $\mathcal{A}_1$ and $\mathcal{A}_2$, one to show the firing guarantee and one to show the non-firing guarantee.
We describe this alternative approach in the next subsection.

The difficulty with using the approach of this subsection seems to arise because our correctness requirements allow some uncertainty in the firing requirements.
Having at least $r_2$-support is supposed to guarantee firing, and not having $r_1$-support is supposed to guarantee non-firing.  
But there is a middle area, in which we have $r_1$-support but do not have $r_2$-support, in which firing is permitted to occur or not occur. 
But our definition of $\leq_{impl}$ requires that the exact firing behavior of the abstract network $\mathcal A$ be emulated in $\mathcal D$.
That seems to be requiring too much, leading us to the alternative approach of the next subsection.

\subsection{Second attempt}

Here is our second attempt at an approach to proving our correctness condition in Definition~\ref{def: recog-ff-multi} using implementation relationships between abstract and detailed networks.
This approach is based on splitting the requirements in Definition~\ref{def: recog-ff-multi} into two conditions, one for firing and one for non-firing.
We define two abstract networks, $\mathcal{A}_1$ and $\mathcal{A}_2$, and show that each satisfies one of the two conditions.
We define two notions of implementation, one preserving the firing guarantees and one preserving the non-firing guarantees.
We show that the same detailed network $\mathcal D$ implements $\mathcal{A}_1$ using the first notion of implementation, and implements $\mathcal{A}_2$ using the second notion of implementation.
Combining the two results gives us the combined correctness condition for $\mathcal D$, i.e., that $\mathcal D$ solves the recognition problem for multi-neuron representations.

We start by defining our two different implementation relationships between a detailed network $\mathcal{D}$ and an abstract network $\mathcal{A}$.
As before, $\mathcal{D}$ includes a fixed set $F$ of failed neurons, satisfying constraints that will be defined in Sections~\ref{sec: detailed-algorithm-high} and~\ref{sec: detailed-algorithm-low}.

\begin{definition}[$Implements_1$, $\leq_{impl1}$]
\label{def: impl-firing}
Network $\mathcal{D}$ $implements_1$ network $\mathcal{A}$ provided that, for every $B \subseteq C_0$, for the two unique executions of these two networks on $B$, the following holds: 
For every concept $c$, if $rep(c)$ fires at time $level(c)$ in $\mathcal{A}$, then in $\mathcal{D}$, at least $m (1-\epsilon)$ of the neurons $v \in reps(c)$ fire, each such $v$ at time $level(c)$.

We write $\mathcal{D} \leq_{impl1} \mathcal{A}$ as shorthand for $\mathcal{D} \ implements_1 \ \mathcal{A}$.
\end{definition}

\begin{definition}[$Implements_2$, $\leq_{impl2}$]
\label{def: impl-non-firing}
Network $\mathcal{D}$ $implements_2$ network $\mathcal{A}$ provided that, for every $B \subseteq C_0$, for the two unique executions of these two networks on $B$, the following holds: 
For every concept $c$, if $rep(c)$ does not fire at time $level(c)$ in $\mathcal{A}$, then in $\mathcal D$, no neuron $v \in reps(c)$ fires at time $level(c)$.

We write $\mathcal{D} \leq_{impl2} \mathcal{A}$ as shorthand for 
$\mathcal{D}$  $implements_2$  $\mathcal{A}$.
\end{definition}

Using the two implementation relationships $\leq_{impl1}$ and $\leq_{impl2}$, we obtain two theorems:

\begin{theorem}
\label{thm: general-mapping-theorem1}
Let $r_2$ be a real in $[0,1]$.
If  $\mathcal{A}$ guarantees $r_2$-firing for $\mathcal C$ according to Definition~\ref{def: firing-guarantee} and $\mathcal{D} \leq_{impl1} \mathcal{A}$, then $\mathcal{D}$ satisfies Part 1 of Definition~\ref{def: recog-ff-multi}.
\end{theorem}

\begin{proof}
Assume that $\mathcal{A}$ guarantees $r_2$-firing for $\mathcal C$ and $\mathcal{D} \leq_{impl1} \mathcal{A}$.
We must show that $\mathcal{D}$ satisfies Part 1 of Definition~\ref{def: recog-ff-multi}.
Fix some particular $B \subseteq C_0$.

Suppose that $c \in supp_{r_2}(B)$. 
Then by Definition~\ref{def: firing-guarantee}, $rep(c)$ fires at time $level(c)$ in $\mathcal A$.
Then Definition~\ref{def: impl-firing} implies that, in $\mathcal D$, at least $m (1-\epsilon)$ of the neurons $v \in reps(c)$ fire, each at time $level(c)$.   This shows Part 1 of Definition~\ref{def: recog-ff-multi}.
\end{proof}

\begin{theorem}
\label{thm: general-mapping-theorem2}
Let $r_1$ be a real in $[0,1]$.  
If $\mathcal{A}$ guarantees $r_1$-non-firing for $\mathcal C$ according to Definition~\ref{def: non-firing-guarantee} and $\mathcal{D} \leq_{impl2} \mathcal{A}$, 
then $\mathcal{D}$ satisfies Part 2 of Definition~\ref{def: recog-ff-multi}.
\end{theorem}

\begin{proof}
Assume that $\mathcal{A}$ guarantees $r_1$-non-firing for $\mathcal C$ and $\mathcal{D} \leq_{impl2} \mathcal{A}$.
We must show that $\mathcal{D}$ satisfies Part 2 of Definition~\ref{def: recog-ff-multi}.
Fix some particular $B \subseteq C_0$.

Suppose that $c \notin supp_{r_1}(B)$.  
Then by Definition~\ref{def: non-firing-guarantee}, $rep(c)$ does not fire at time $level(c)$ in $\mathcal{A}$.
Then Definition~\ref{def: impl-non-firing} implies that, in $\mathcal D$, no neuron in $reps(c)$ fires at time $level(c)$.  This shows Part 2 of Definition~\ref{def: recog-ff-multi}.
\end{proof}

Now we consider two abstract networks $\mathcal{A}_1$ and $\mathcal{A}_2$, where $\mathcal{A}_1$ guarantees $r_2$-firing for $\mathcal C$ and $\mathcal{A}_2$ guarantees $r_1$-non-firing for $\mathcal C$.
Combining the two previous results, we get:

\begin{theorem}[\textbf{Combined mapping theorem}]
\label{thm:  general-mapping-theorem3}
Suppose that network $\mathcal{A}_1$ guarantees $r_2$-firing for $\mathcal C$ and network
$\mathcal{A}_2$ guarantees $r_1$-non-firing for $\mathcal C$.
Suppose that $\mathcal D \leq_{impl1} \mathcal{A}_1$ and $\mathcal D \leq_{impl2} \mathcal{A}_2$.
Then $\mathcal D$ $(r_1,r_2)$-recognizes $\mathcal C$ according to Definition~\ref{def: recog-ff-multi}.
\end{theorem}

\begin{proof}
Immediate from Theorems~\ref{thm: general-mapping-theorem1} and~\ref{thm: general-mapping-theorem2}. 
\end{proof}

\section{Abstract Networks}
\label{sec: abstract-algorithms}

In this section, we define two abstract networks, $\mathcal{A}_1$ and $\mathcal{A}_2$.  They are identical except for different firing thresholds.
In Sections~\ref{sec: detailed-algorithm-high} and~\ref{sec: detailed-algorithm-low}, $\mathcal{A}_1$ and $\mathcal{A}_2$ will serve as abstract networks for two detailed networks, $\mathcal H$ and $\mathcal L$, with high connectivity and low connectivity respectively.
We will use $\mathcal{A}_1$ for proving firing guarantees and $\mathcal{A}_2$ for non-firing guarantees.

We assume that, in $\mathcal{A}_1$ and $\mathcal{A}_2$, $\ell'_{max} = \ell_{max}$.
Also, every concept $c \in C$ has one representing neuron, $rep(c)$, as described in Section~\ref{sec: recog-single}, with $layer(rep(c)) = level(c)$.
The $rep$ functions may be different for $\mathcal{A}_1$ and $\mathcal{A}_2$.  
$\mathcal {A}_1$ and $\mathcal{A}_2$ do not admit any neuron failures.

For each of these networks, we assume:
\begin{itemize}
    \item $r_1 \leq r_2$.
    \item  For each neuron $v$ of the form $rep(c)$ with $level(c) \geq 1$, there are weight $1$ edges from all $reps$ of children of $c$ to $v$.  All other edges have weight $0$.     
\end{itemize}
The assumption that $r_1 \leq r_2$ is sufficient to argue correctness of the two abstract networks $\mathcal{A}_1$ and $\mathcal{A}_2$ on their own.
Later, in order to use them as abstractions of the detailed networks $\mathcal H$ and $\mathcal L$, we will need something stronger:  an appropriate gap between $r_1$ and $r_2$ to accommodate failures and limited connectivity.\footnote{
For network $\mathcal H$, we must take account of failures, but not of disconnections.
So we will assume there that $r_1 \leq r_2 (1 - \epsilon)$.
For network $\mathcal L$, we will need a greater gap, $r_1 \leq a r_2 (1 - \epsilon)$, where $a$ is some constant $< 1$, to handle limited connectivity as well as neuron failures.
}
We assume weight $1$ edges from $reps$ of children to $reps$ of their parents.

\subsection{Abstract Network $\mathcal{A}_1$}

For $\mathcal{A}_1$, we make the additional assumption:
\begin{itemize}
    \item $\tau = r_2 k$.  
\end{itemize}
Thus, we assume a fairly high firing threshold.

Recognition in network $\mathcal{A}_1$ proceeds as in~\cite{LM21}.
That is, we start by presenting a set $B \subseteq C_0$ at time $0$.\footnote{This requires that layer $0$ contain at least $|C_0|$ neurons.  From now on, we will just assume that we have enough neurons to represent all the concepts.} 
More precisely, we assume that the $rep$ of each level $0$ concept in $B$ fires at time $0$, and only at time $0$, and no other layer $0$ neurons ever fire.  This leaves us with a deterministic system.
Then the neuron firing propagates up the layers, one time unit per layer, according to the deterministic firing rule defined in Section~\ref{sec: network-operation}.

The global states of $\mathcal{A}_1$ consist of the following components:
\begin{itemize}
    \item 
    For each neuron $u$ in the network, $firing^u \in \{ 0, 1\}$.  Initially, the layer $0$ neurons' $firing$ components are set according to the definition of presenting $B$ for single-neuron representations.  All the higher-layer neurons have their initial $firing$ components set to $0$.
    \item  
    For each neuron $v$ with $layer(v) = \ell \geq 1$, a mapping $weight^v$ from neurons at layer $\ell-1$ to $\{0,1\}$.
    For each neuron $v$ of the form $rep(c)$, we have weight $1$ edges from all $reps$ of children of $c$ to $v$.  All other edges have weight $0$.  That is, $weight^v(u) = 1$ exactly if $u = rep(c')$ for some child $c'$ of $c$.
\end{itemize} 
We omit the $failed$ components, since we are not considering failures here.

The transitions are also global.  
The $weight$ components don't change.
For the $firing$ flags, the values for the layer $0$ neurons are $0$ after any transition, since we are assuming that inputs are presented only at time $0$.
The value for each higher-layer neuron $v$ after a transition is determined based on the incoming potential:  if that is at least $r_2 k$, then the new value of $firing^v$ is $1$, and otherwise it is $0$.

\begin{theorem}
\label{thm: A1-firing-guarantee}
$\mathcal{A}_1$ guarantees $r_2$-firing for $\mathcal C$, as in Definition~\ref{def: firing-guarantee}.
\end{theorem}

\begin{proof} (Sketch)
We can argue correctness as in~\cite{LM21,DBLP:conf/sirocco/LynchM23}.
Here we need to show that, for any concept $c$ in $\mathcal C$ that is $r_2$-supported by the presented set $B$, the neuron $rep(c)$ fires at time $level(c)$.  
This can be shown using an inductive proof on levels.

The key is the inductive step, where we consider a concept $c$ at level $\ell \geq 1$ that is $r_2$-supported by $B$.
Then by the definition of $r_2$-supported, it must be that at least $r_2 k$ of $c$'s children are $r_2$-supported by $B$.
By the inductive hypothesis, all of the $reps$ of these children fire at time $\ell-1$.
Since the edges from these $reps$ to $rep(c)$ all have weight $1$, the total incoming potential to $rep(c)$ is at least $r_2 k$.
This meets the firing threshold $r_2 k$ for $rep(c)$ and ensures that $rep(c)$ fires at time $\ell$.
\end{proof}

In fact, we can show the following stronger two-directional result, though we will not use it in this paper.  The proof is as in~\cite{LM21,DBLP:conf/sirocco/LynchM23}.

\begin{theorem}
\label{thm: A1-both-dirs}
$\mathcal{A}_1$ $(r_1,r_2)$-recognizes $\mathcal C$, as in Definition~\ref{def: recog-ff-single}.
\end{theorem}

\subsection{Abstract Network $\mathcal{A}_2$}

This is exactly the same as $\mathcal{A}_1$, but with a somewhat lower firing threshold:
\begin{itemize}
    \item $\tau = r_1 k$.
\end{itemize}

The network operates in the same way as $\mathcal{A}_1$. 
That is, we start by presenting a set $B \subseteq C_0$ at time $0$.  
More precisely, we assume that the $rep$ of each level $0$ concept in $B$ fires at time $0$, and only at time $0$, and no other layer $0$ neurons ever fire.  This leaves us with a deterministic system.
Then the neuron firing propagates up the layers, one time unit per layer, according to the deterministic firing rule defined in Section~\ref{sec: network-operation}.

The global states again consist of:
\begin{itemize}
    \item 
    For each neuron $u$ in the network, $firing^u \in \{ 0, 1\}$.  Initially, the layer $0$ neurons' $firing$ components are set according to the definition of presenting $B$ for single-neuron representations.  All the higher-layer neurons have their initial $firing$ components set to $0$.
    \item  
    For each neuron $v$ with $layer(v) = \ell \geq 1$, a mapping $weight^v$ from neurons at layer $\ell-1$ to $\{0,1\}$.
    For each neuron $v$ of the form $rep(c)$, we have weight $1$ edges from all  $reps$ of children of $c$ to $v$.  All other edges have weight $0$.  That is, $weight^v(u) = 1$ exactly if $u = rep(c')$ for some child $c'$ of $c$.
\end{itemize} 
We again omit the $failed$ components.
The transitions are exactly as for $\mathcal{A}_1$. 

\begin{theorem}
\label{thm: A2-non-firing-guarantee}
Network $\mathcal{A}_2$ guarantees $r_1$-non-firing for $\mathcal C$, as in Definition~\ref{def: non-firing-guarantee}.
\end{theorem}

\begin{proof} (Sketch)
Now we must show that, for any concept $c \in \mathcal{C}$ that is not $r_1$-supported by $B$, the neuron $rep(c)$ does not fire at time $level(c)$.
Again this can be shown using induction on levels.

For the inductive step, we consider a concept $c$ at level $\ell \geq 1$ that is not $r_1$-supported by $B$.  
Then it must be that strictly fewer than $r_1 k$ of $c$'s children are $r_1$-supported by $B$.
By the inductive hypothesis, the $reps$ of only these children could fire at time $\ell - 1$.
Therefore, the total incoming potential to $rep(c)$ is strictly less than $r_1 k$.
This does not meet the firing threshold $r_1 k$ for $rep(c)$, which ensures that $rep(c)$ does not fire at time $\ell$.
\end{proof}

Again, we can show the stronger two-directional result, though we will not use it in this paper. 
The proof is as in~\cite{LM21,DBLP:conf/sirocco/LynchM23}.

\begin{theorem}
\label{thm: A2-both-dirs}
$\mathcal{A}_2$ $(r_1,r_2)$-recognizes $\mathcal C$, as in Definition~\ref{def: recog-ff-single}.
\end{theorem}


Theorems~\ref{thm: A1-both-dirs} and~\ref{thm: A2-both-dirs} say that each of $\mathcal{A}_1$ and $\mathcal{A}_2$ satisfies both parts of the correctness condition for recognition, in Definition~\ref{def: recog-ff-single}  .
In fact, we have enough freedom in setting the thresholds so that any threshold in the range $[r_1 k, r_2 k]$ would also work correctly, for both parts.

The reason we separate these two conditions here is that we will use the separation in showing correctness for our detailed networks $\mathcal H$ and $\mathcal L$.
For each of these, we will use two separate implementation relationships, mapping to $\mathcal{A}_1$ via $\leq_{impl1}$ to show the firing condition and mapping to $\mathcal{A}_2$ via $\leq_{impl2}$ to show the non-firing condition.\footnote{It might be possible to define a single nondeterministic abstract network that can be used for both parts, but that seems more complicated and we haven't done this.}

\section{Detailed Network for a Network Model with High Connectivity}
\label{sec: detailed-algorithm-high}

We are finally ready to describe our first detailed network, $\mathcal H$.  It is based on a feed-forward network model with total connectivity from each layer $\ell-1$ to layer $\ell$.
In $\mathcal H$, $\ell'_{max} = \ell_{max}$.  
Also, every concept $c \in C$ has exactly $m$ representing neurons, $reps(c)$, as described in Section~\ref{sec: recog-multi}, and every neuron $v \in reps(c)$ has $layer(v) = level(c)$.
Network $\mathcal H$ includes a set $F$ of failed neurons, subject to constraints that we define below.

Network $\mathcal H$ is a deterministic version of the first probabilistic network in~\cite{lynch2024multineuron}.  
In~\cite{lynch2024multineuron}, we gave a direct proof that the network of that paper works correctly, that is, that it $(r_1,r_2)$ recognizes $\cal C$.
The definition of $(r_1,r_2)$ recognition in that paper is probabilistic, asserting recognition only with high probability.\footnote{Also, as noted earlier, in~\cite{lynch2024multineuron} we assumed a variant of our concept hierarchy definition, with only one root.  The proofs in that paper could be extended to the more general case, with suitably modified probability bounds.}

In this paper, we give a different style of proof for $\mathcal H$, based on implementation relationships $\leq_{impl1}$ and $\leq_{impl2}$ from $\mathcal H$ to our abstract networks $\mathcal{A}_1$ and $\mathcal{A}_2$, respectively.
We know from Theorem~\ref{thm: A1-firing-guarantee} that $\mathcal{A}_1$ guarantees $r_2$-firing for $\mathcal C$, and from Theorem~\ref{thm: A2-non-firing-guarantee} that $\mathcal{A}_2$ guarantees $r_1$-non-firing for $\mathcal C$.
We argue here that the implementations carry over these correctness properties to the detailed network $\mathcal H$.
This argument is based on the combined mapping theorem, Theorem~\ref{thm:  general-mapping-theorem3}.

The first thing we do here, compared to~\cite{lynch2024multineuron}, is to remove all the probabilistic choices from the network, in order to produce a deterministic network.
The only probability in the network of~\cite{lynch2024multineuron} is in the choice of the set of neurons that fail:  each neuron fails independently, with the same probability $q$.
Here we assume simply that some arbitrary set $F$ of neurons fail, subject to some constraints, namely, that the number of surviving (i.e., non-failing) neurons among the $reps$ of any $c$ is at least $m (1 - \epsilon)$, where $\epsilon$ is a recognition approximation parameter.
The first "survival lemma" in~\cite{lynch2024multineuron}, Lemma 5.2, says that this bound is achieved in the network of~\cite{lynch2024multineuron} with high probability.

Thus, we rely on a survival lemma from~\cite{lynch2024multineuron} to argue that the probabilistic version of the network satisfies our desired constraints on the set $F$ of failed neurons.
Here, we just assume the constraints and forget about the probability.
Once we fix $F$, the behavior of the network is completely determined by the input set $B$.

\subsection{The network $\mathcal{H}$}

We use the following assumptions:
\begin{itemize}
\item 
$\tau = r_2 k m (1- \epsilon)$.
That is, compared to what we used for $\mathcal{A}_1$, the threshold now includes the number $m$ of $reps$ of each concept, as well as a factor of $1 - \epsilon$ that is used in bounding the number of failures.\footnote{To make our results a bit more general, we could use an arbitrary threshold in the range ${r_1 k m (1 - \epsilon), r_2 k m (1 - \epsilon)]}$ instead of just the higher threshold $r_2 k m (1 - \epsilon)$.  The proof would be mostly unchanged.}
\item 
$r_1 \leq r_2 (1 - \epsilon)$.
\item 
For each neuron $v$ of the form $rep(c)$ with $level(c) \geq 1$, there are weight $1$ edges from all $reps$ of children of $c$ to $v$.  All other edges have weight $0$. 
\item 
$F$ is the set of failed neurons.  We assume that the number of surviving (i.e., non-failed) neurons among the $reps$ of any $c$ is at least $m (1 - \epsilon)$.
For non-$rep$ neurons, we have no constraints on their failures, but we use the fixed set $F$ in order to ensure that the network's behavior is deterministic, for a given input $B$.
\end{itemize}
Thus, we assume a fairly high firing threshold and limited failures.
Now we assume a gap between $r_1$ and $r_2$, namely, $r_1 \leq r_2 (1 - \epsilon)$, in order to accommodate neuron failures.
We assume weight $1$ edges between $reps$ of children and $reps$ of their parents.

The network operates as in~\cite{lynch2024multineuron}.
That is, we start by presenting a set $B \subseteq C_0$ at time $0$.
More precisely, we assume that all surviving $reps$ of each level $0$ concept in $B$ fire at time $0$, and only at time $0$, and no other layer $0$ neurons fire at any time.  This leaves us with a deterministic system.
Then the neuron firing propagates up the layers, one time unit per layer, according to the deterministic firing rule in Section~\ref{sec: network-operation}.

The global states of $\mathcal H$ consist of the following components:
\begin{itemize}
    \item
    For each neuron $u$ in the network, $failed^u \in \{0,1\}$. 
    The initial values are set according to $F$:  $failed^u = 1$ if and only if $u \in F$.
    \item 
    For each neuron $u$ in the network, $firing^u \in \{ 0, 1 \}$.  
    Initially, the layer $0$ neurons’ $firing$ components are set according to the definition of presenting $B$ for multi-neuron representations. 
    That is, if $u \in reps(c) - F$ then $u$ fires at time $0$, whereas 
    if $u \notin reps(c)$ or $u \in F$, then $u$ does not fire at time $0$.
    All the higher-layer neurons have their initial $firing$ components set to $0$.
    \item 
    For each neuron $v$ with $layer(v) = \ell \geq 1$, a mapping $weight^v$ from neurons at layer $\ell-1$ to $\{0,1\}$. For each neuron $v \in reps(c)$, we have weight $1$ edges from all $reps$ of children of $c$ to $v$.  All other edges have weight $0$.  That is, $weight^v(u) = 1$ exactly if $u \in reps(c')$ for some child $c'$ of $c$.
\end{itemize}

The $firing$ and $weight$ components are analogous to the corresponding components of $\mathcal{A}_1$ and $\mathcal{A}_2$, except that now we consider multiple $reps$ for each concept $c$.  In particular, we have weight-$1$ connectivity between all of the $reps$ of concepts and $reps$ of their children.
The $failed$ component is new.

As for the transitions:
The $weight$ and $failed$ components don't change.  
For the $firing$ flags, the values for the layer $0$ neurons are $0$ after any transition.
The $firing$ values of higher-layer neurons are computed based on potential, with a sharp threshold of $r_2 k m (1-\epsilon)$.  The failed neurons do not contribute to this potential.

\subsection{$\mathcal{H}$ implements $\mathcal{A}_1$}

In this section, we establish an implementation relationship between the detailed network $\mathcal H$ and the abstract network $\mathcal{A}_1$.
We use this to carry over firing guarantees from $\mathcal{A}_1$ to $\mathcal H$, in the manner described by Theorem~\ref{thm: general-mapping-theorem1}.

We assume a fixed set $F$ of failed nodes in $\mathcal H$, subject to the constraint that the number of surviving neurons among the $reps$ of any $c$ is at least $m (1 - \epsilon)$.
We show:

\begin{theorem}
\label{thm: H-firing-guarantee}
$\mathcal{H} \leq_{impl1} \mathcal{A}_1$.  
\end{theorem}

In other words, for every $B \subseteq C_0$, for the two unique executions of the two networks $\mathcal{A}_1$ and $\mathcal H$ on $B$, the following holds: 
For every concept $c$, if $rep(c)$ fires at time $level(c)$ in $\mathcal{A}_1$, then in $\mathcal{H}$, at least $m (1-\epsilon)$ of the neurons in $reps(c)$ fire at time $level(c)$.

\begin{proof}
Fix $B \subseteq C_0$, and consider the unique executions of $\mathcal{A}_1$ and $\mathcal H$ on input $B$.
We prove the following statement $P(t)$, for every $t \geq 0$.  This immediately implies the result.

$P(t)$:  For every concept $c$ with $level(c) = t$, if $rep(c)$ fires at time $t$ in $\mathcal{A}_1$, then in $\mathcal H$, at least $m (1-\epsilon)$ of the neurons in $reps(c)$ fire at time $t$.

We prove $P(t)$ by induction on $t$, for all concepts $c$.

\emph{Base:}  $t = 0$. 
Suppose that $c$ is a concept with $level(c) = 0$ and $rep(c)$ fires at time $0$ in network $\mathcal{A}_1$.
Then by definition of presentation in $\mathcal{A}_1$, $c \in B$.
Then by definition of presentation in $\mathcal{H}$ and the assumption on $\mathcal H$ that at least $m (1 - \epsilon)$ of the neurons in $reps(c)$ survive, at least $m (1 - \epsilon)$ of the neurons in $reps(c)$ fire at time $0$ in $\mathcal H$, as needed.

\emph{Inductive step:} $t \geq 1$:
Assume that $P(t-1)$ holds.
Consider $c$ with $level(c) = t$.
Suppose that $rep(c)$ fires at time $t$ in $\mathcal{A}_1$.
Then $rep(c)$ must have enough incoming potential from the $reps$ of its children to fire.
Since the threshold in $\mathcal{A}_1$ is $r_2 k$, this means that the $reps$ of at least $r_2 k$ of $c$'s children fire at time $t-1$.

Then by the inductive hypothesis $P(t-1)$, we have that, in $\mathcal{H}$, for each of these $r_2 k$ children $c'$ of $c$, at least $m (1 - \epsilon)$ of the neurons in $reps(c')$ fire at time $t-1$.  This yields a total incoming potential to each neuron $v \in reps(c)$ of at least $r_2 k m (1-\epsilon)$. 
This is enough to meet the firing threshold for $v$.
By assumption on $\mathcal H$, at least $m (1 - \epsilon)$ of the neurons $v \in reps(c)$ survive.  Therefore, since their thresholds are met, they fire.
Thus, at least $m (1 - \epsilon)$ of the neurons in $reps(c)$ fire at time $t$ in $\mathcal H$, as needed.
\end{proof}

Now we can show that $\mathcal H$ satisfies the firing requirement.

\begin{theorem}
\label{thm: H-firing-guarantee-2}
$\mathcal{H}$ satisfies Part 1 of the definition of the recognition problem for networks with multi-neuron representations, Definition~\ref{def: recog-ff-multi}.
\end{theorem}

\begin{proof}
Theorem~\ref{thm: A1-firing-guarantee} says that $\mathcal{A}_1$ guarantees $r_2$-firing for $\mathcal{C}$.  
Theorem~\ref{thm: H-firing-guarantee} says that $\mathcal{H} \leq_{impl1} \mathcal{A}_1$.
Then Theorem~\ref{thm: general-mapping-theorem1} implies that $\mathcal{H}$ satisfies Part 1 of Definition~\ref{def: recog-ff-multi}.
\end{proof}

\subsection{$\mathcal{H}$ implements $\mathcal{A}_2$}

Now we establish a second implementation relationship, between the detailed network $\mathcal H$ and the abstract network $\mathcal{A}_2$, in order to carry over non-firing guarantees from $\mathcal{A}_2$ to $\mathcal{H}$, in the manner described by Theorem~\ref{thm: general-mapping-theorem2}.

Again we assume a fixed set $F$ of failed nodes, subject to the constraint that the number of surviving neurons among the $reps$ of any $c$ is at least $m (1 - \epsilon)$.  Now we show:

\begin{theorem}
\label{thm: H-non-firing-guarantee}
$\mathcal{H} \leq_{impl2} \mathcal{A}_2$.
\end{theorem}

In other words, for every $B \subseteq C_0$, for the two unique executions of the two networks $\mathcal{A}_2$ and $\mathcal H$ on $B$, the following holds:  For every concept $c$, if $rep(c)$ does not fire at time $level(c)$ in $\mathcal{A}_2$, then in $\mathcal H$, no neurons in $reps(c)$ fire at time $level(c)$.  

\begin{proof}
Fix $B \subseteq C_0$, and consider the unique executions of $\mathcal{A}_2$ and $\mathcal H$ on input $B$.
We prove the following statement $P(t)$, for every $t \geq 0$, which immediately implies the result.

$P(t)$:
For every concept $c$ with $level(c) = t$, if $rep(c)$ does not fire at time $t$ in $\mathcal{A}_2$ then in $\mathcal H$, no neurons in $reps(c)$ fire at time $t$.

We prove $P(t)$ by induction on $t$, for all concepts $c$.

\emph{Base:}  $t = 0$: 
Suppose that $c$ is a concept with $level(c) = 0$  and $rep(c)$ does not fire at time $0$ in network $\mathcal{A}_2$.
Then by definition of presentation in $\mathcal{A}_2$, $c \notin B$. 
Then by definition of presentation in $\mathcal H$, none of the neurons in $reps(c)$ fire at time $0$, as needed.

\emph{Inductive step:}  $t \geq 1$: 
Assume that $P(t-1)$ holds, and consider $c$ with $level(c) = t$.
Suppose that $rep(c)$ does not fire at time $t$ in $\mathcal{A}_2$.
Then $rep(c)$ must not have enough incoming potential from its children to fire.
Since the threshold in $\mathcal{A}_2$ is $r_1 k$, this means that the $reps$ of strictly fewer than $r_1 k$ of $c$'s children fire at time $t-1$.

Then by inductive hypothesis $P(t-1)$, we have that, in $\mathcal H$, for each child $c'$ of $c$ that does not fire, none of the neurons in $reps(c')$ fire at time $t-1$.
This yields a total incoming potential to each neuron $v \in reps(c)$ of strictly less than $r_1 k m$.
But the firing threshold for each such $v$ is $r_2 k m (1-\epsilon)$.
Since we have assumed that $r_1 \leq r_2 (1 - \epsilon)$,
we know that $r_1 k m \leq r_2 k m (1 - \epsilon)$.
So the incoming potential to each such $v$ is strictly less than its firing threshold $r_2 k m (1-\epsilon)$, which implies that $v$ does not fire at time $t$, as needed.
\end{proof}

Now we can show that $\mathcal H$ satisfies the non-firing requirement.

\begin{theorem}
\label{thm: H-non-firing-guarantee-2}
$\mathcal{H}$ satisfies Part 2 of the definition of the recognition problem for networks with multi-neuron representations, Definition~\ref{def: recog-ff-multi}.
\end{theorem}

\begin{proof}
Theorem~\ref{thm: A2-non-firing-guarantee} says that $\mathcal{A}_2$ guarantees $r_1$-non-firing for $\mathcal{C}$.  
Theorem~\ref{thm: H-non-firing-guarantee} says that $\mathcal{H} \leq_{impl2} \mathcal{A}_2$.
Then Theorem~\ref{thm: general-mapping-theorem2} implies that $\mathcal{H}$ satisfies Part 2 of Definition~\ref{def: recog-ff-multi}.
\end{proof}

\subsection{Correctness of $\mathcal H$}

\begin{theorem}
$\mathcal{H}$ $(r_1,r_2)$-recognizes $\mathcal C$.
\end{theorem}

\begin{proof}
By Theorems~\ref{thm: H-firing-guarantee-2} and~\ref{thm: H-non-firing-guarantee-2}.
\end{proof}

Thus, we have shown correctness of the deterministic network $\mathcal{H}$.  The network of~\cite{lynch2024multineuron} is probabilistic, in that it chooses the failed set $F$ randomly.
It is shown in~\cite{lynch2024multineuron} that, with high probability, that network achieves the restrictions on failures that are assumed here.
Therefore, we can infer a claim about the probabilistic version of the network:  that with high probability it $(r_1,r_2)$-recognizes $\mathcal C$.\footnote{Again, we remark that the probabilities here are slightly larger than in~\cite{lynch2024multineuron}, because the "shapes" of the hierarchies are slightly different.} 

\section{Detailed Network for a Network Model with Low Connectivity}
\label{sec: detailed-algorithm-low}

Now we add the complication of limited connectivity, in addition to neuron failures.
We describe our second detailed network, $\mathcal{L}$.
$\mathcal{L}$ is also based on a feed-forward network, but now with partial connectivity from each layer $\ell - 1$ to layer $\ell$.
In $\mathcal L$, as in $\mathcal H$, $\ell'_{max} = \ell_{max}$.  
Also, every concept $c \in C$ has exactly $m$ representing neurons, $reps(c)$, as described in Section~\ref{sec: recog-multi}, and every neuron $v \in reps(c)$ has $layer(v) = level(c)$.
Network $\mathcal L$ includes a set $F$ of failed neurons, subject to constraints that we define below.

Network $\mathcal{L}$ is a deterministic version of the second probabilistic network in~\cite{lynch2024multineuron}.
In~\cite{lynch2024multineuron}, we gave a direct proof that the network of that paper works correctly, that is, that it $(r_1,r_2)$-recognizes $\mathcal C$, using a probabilistic version of the definition.
In this paper, we give a new proof based on implementation relationships $\leq_{impl1}$ and $\leq_{impl2}$ from $\mathcal L$ to our abstract networks $\mathcal{A}_1$ and $\mathcal{A}_2$, respectively.

Again, we begin by removing all of the probabilistic choices from the network, in order to produce a deterministic network.
Again, the only probability in the network of~\cite{lynch2024multineuron} is in the choice of the set $F$ of failed neurons.  The connectivity of the algorithm in~\cite{lynch2024multineuron} is not determined probabilistically, but rather, nondeterministically subject to some constraints.  
Here we assume that some arbitrary set $F$ of neurons fails, subject to the same constraint as for $\mathcal H$, that is, that the number of surviving neurons among the $reps$ of any $c$ is at least $m (1-\epsilon)$.
We also add a new constraint: that for any concept $c$ with $level(c) \geq 1$, any $v \in reps(c)$, and any child $c'$ of $c$, there are at least $a m (1 - \epsilon)$ neurons in $reps(c')$ that both survive and are connected with weight $1$ edges to $v$.
Two of the "survival lemmas" in~\cite{lynch2024multineuron}, Lemmas 6.2 and 6.4, say that these bounds are achieved in the network of~\cite{lynch2024multineuron} with high probability.

Thus, we again rely on survival lemmas from~\cite{lynch2024multineuron} to argue that the probabilistic version of the network satisfies our desired constraints.  Here, we just assume the constraints and forget about the probability. 
Once we fix $F$ and the edge weights, the behavior of the network is completely determined by the input set $B$.


\subsection{The network $\mathcal{L}$}
\label{sec: L-def}

Now we use the following assumptions:
\begin{itemize}
\item 
$\tau = a r_2 k m (1-\epsilon)$.\footnote{As before, to make our results a bit more general, we could use an arbitrary threshold in the range ${a r_1 k m (1 - \epsilon), a r_2 k m (1 - \epsilon)]}$ instead of just the higher threshold $a r_2 k m (1 - \epsilon)$.}
\item 
$r_1 \leq a r_2 (1 - \epsilon)$.
\item 
$F$ is the set of failed neurons.  We assume that the number of surviving neurons among the $reps$ of any $c$ is at least $m (1 - \epsilon)$. 
\item 
$E$ is the set of $weight$ $1$ edges from $reps$ of children to $reps$ of their parents.  These represent the (partial) connectivity of the network.
We assume that, for any concept $c$ with $level(c) \geq 1$, any $v \in reps(c)$, and any child $c'$ of $c$, there are at least $a m (1 - \epsilon)$ neurons $u \in reps(c')$ such that $u$ survives and $(u,v) \in E$.
\end{itemize}
Thus, we assume a fairly low firing threshold and limited failures.
Now we assume a larger gap between $r_1$ and $r_2$ than in $\mathcal H$, namely, $r_1 \leq a r_2 (1 - \epsilon)$, in order to accommodate both neuron failures and limited connectivity.
We assume weight $1$ edges between "sufficiently many" surviving $reps$ of children and $reps$ of their parents.

The network operates as in~\cite{lynch2024multineuron}.
That is, we start by presenting a set $B \subseteq C_0$ at time $0$.
Then the neuron firing propagates up the layers, according to the deterministic firing rule in Section~\ref{sec: network-operation}.

The global states of $\mathcal L$ consist of the following components, which are the same as for $\mathcal H$.  The only difference is in the weights of edges, reflecting partial connectivity.
\begin{itemize}
    \item
    For each neuron $u$ in the network, $failed^u \in \{0,1\}$. 
    The initial values are set according to $F$:  $failed^u = 1$ if and only if $u \in F$.
    \item 
    For each neuron $u$ in the network, $firing^u \in \{ 0, 1 \}$.  Initially, the layer $0$ neurons’ $firing$ components are set according to the definition of presenting $B$ for multi-neuron representations. 
    That is, if $u \in reps(c) - F$ then $u$ fires at time $0$, whereas if $u \notin reps(c)$ or $u \in F$, then $u$ does not fire at time $0$.
    All the higher-layer neurons have their initial $firing$ components set to $0$.
    \item 
    For each neuron $v$ with $layer(v) = \ell \geq 1$, a mapping $weight^v$ from neurons at layer $\ell-1$ to $\{0,1\}$.  For each neuron $v \in reps(c)$, we have weight $1$ edges to $v$ from all $u \in reps(children(c))$ such that $(u,v) \in E$.  All other edges have weight $0$.  
\end{itemize}
The transitions are as for $\mathcal H$.
Now only the surviving neurons that are connected to a neuron $v$ via edges in $E$ contribute to the incoming potential to $v$.

\subsection{$\mathcal{L}$ implements $\mathcal{A}_1$}

Now we establish an implementation relationship between the detailed network $\mathcal L$ and the abstract network $\mathcal{A}_1$.
We use this to carry over firing guarantees from $\mathcal{A}_1$ to $\mathcal L$, as described by Theorem~\ref{thm: general-mapping-theorem1}.

We assume a fixed set $F$ of failed nodes in $\mathcal L$, and a fixed set $E$ of weight $1$ edges, subject to the constraints given in Section~\ref{sec: L-def}.  We show:

\begin{theorem}
\label{thm: L-firing-guarantee}
$\mathcal{L} \leq_{impl1} \mathcal{A}_1$.
\end{theorem}

In other words, for every $B \subseteq C_0$, for the two unique xecutions of the two networks $\mathcal{A}_1$ and $\mathcal L$ on $B$, the following holds:
For every concept $c$, if $rep(c)$ fires at time $level(c)$ in $\mathcal{A}_1$, then in $\mathcal L$, at least $m (1 - \epsilon)$ of the neurons in $reps(c)$ fire at time $level(c)$.

In order to prove the corresponding theorem for $\mathcal H$, Theorem~\ref{thm: H-firing-guarantee}, we used induction on $t$ to prove the following statement $P(t)$:
For every concept $c$ with $level(c) = t$, if $rep(c)$ fires at time $t$ in $\mathcal{A}_1$, then in $\mathcal H$, at least $m (1 - \epsilon)$ of the neurons in $reps(c)$ fire at time $t$.
However, for $\mathcal L$, the combination of partial connectivity and failures adds new complications, which require us to state a stronger predicate to be proved by induction.

\begin{proof}
Fix $B \subseteq C_0$, and consider the unique executions of $\mathcal{A}_1$ and $\mathcal L$ on input $B$.
We prove the following statement $Q(t)$, for every $t \geq 0$:

$Q(t)$:  For every concept $c$ with $level(c) = t$, if $rep(c)$ fires at time $t$ in $\mathcal{A}_1$, then in $\mathcal L$, every surviving neuron in $reps(c)$ fires at time $t$. 

By our assumptions on $F$, we know that, for any $c$, the number of surviving neurons in $reps(c)$ is at least $m (1 - \epsilon)$.
Therefore, $Q(t)$ implies:  
For every concept $c$ with $level(c) = t$, if $rep(c)$ fires at time $t$ in $\mathcal{A}_1$, then in $\mathcal L$,  at least $m (1 - \epsilon)$ of the neurons in $reps(c)$ fire at time $t$.
This is what is needed for the implementation relationship $\leq_{impl1}$.

We prove $Q(t)$ by induction on $t$, for all concepts $c$.

\emph{Base:} $t = 0$:
Suppose that $c$ is a concept with $level(c) = 0$ and $rep(c)$ fires at time $0$ in $\mathcal{A}_1$.
Then by definition of presentation in $A_1$, $c$ is in $B$.
Then by definition of presentation in $L$, all surviving neurons in $reps(c)$ fire, as needed.

\emph{Inductive step:} $t \geq 1$: 
Assume that $Q(t-1)$ holds, and consider $c$ with $level(c) = t$.
Suppose that $rep(c)$ fires at time $t$ in $\mathcal{A}_1$.
Then $rep(c)$ must have enough incoming potential from the $reps$ of its children to fire.
Since the threshold in $\mathcal{A}_1$ is $r_2 k$, it means that the $reps$ of at least $r_2 k$ of $c$'s children fire at time $t-1$.
Let $C'$ be the set of children of $c$ whose $reps$ fire at time $t-1$, in $\mathcal{A}_1$, so $|C'| \geq r_2 k$.
Then the inductive hypothesis $Q(t-1)$ implies that, in $\mathcal{L}$, for each $c' \in C'$, every surviving neuron in $reps(c')$ fires at time $t-1$.

Now consider any particular surviving neuron $v \in reps(c)$; we show that $v$ fires at time $t$.
Consider any particular $c' \in C'$. 
By an assumption on $\mathcal{L}$, we know that there are at least $a m (1-\epsilon)$ neurons in $reps(c')$ that both survive and are connected to $v$.
Since these neurons all survive, by the inductive hypothesis $Q(t-1)$, all of these neurons fire at time $t-1$.
So this entire collection of neurons in $reps(c')$ contributes at least $a m (1-\epsilon)$ potential to $v$.
Considering all the concepts $c' \in C'$ together, we get a total incoming potential to $v$ of at least $a r_2 k m (1 - \epsilon)$, which is enough to meet the firing threshold for $v$.
Therefore, since $v$ survives, it fires at time $t$, as needed.
\end{proof}

Now we can show that $\mathcal L$ satisfies the firing requirement. 

\begin{theorem}
\label{thm: L-firing-guarantee-2}
$\mathcal{L}$ satisfies Part 1 of the definition of the recognition problem for networks with multi-neuron representations, Definition~\ref{def: recog-ff-multi}.
\end{theorem}

\begin{proof}
Theorem~\ref{thm: A1-firing-guarantee} says that $\mathcal{A}_1$ guarantees $r_2$-firing for $\mathcal{C}$.  
Theorem~\ref{thm: L-firing-guarantee} says that $\mathcal{L} \leq_{impl1} \mathcal{A}_1$.
Then Theorem~\ref{thm: general-mapping-theorem1} implies that $\mathcal{L}$ satisfies Part 1 of Definition~\ref{def: recog-ff-multi}.
\end{proof}

\subsection{$\mathcal{L}$ implements $\mathcal{A}_2$}

Now we establish one more implementation relationship, between the detailed network $\mathcal{L}$ and the abstract network $\mathcal{A}_2$, in order to carry over the non-firing guarantees from $\mathcal{A}_2$ to $\mathcal L$.
We assume fixed sets $F$ of failed nodes and $E$ of weight $1$ edges, subject to the constraints given in Section~\ref{sec: L-def}.  We show:

\begin{theorem}
\label{thm: L-non-firing-guarantee}
$\mathcal{L} \leq_{impl2} \mathcal{A}_2$.
\end{theorem}

In other words, for every $B \subseteq C_0$, for the two unique executions of the two networks $\mathcal{A}_2$ and $\mathcal L$ on $B$, the following holds:  
For every concept $c$, if $rep(c)$ does not fire at time $level(c)$ in $\mathcal{A}_2$, then in $\mathcal{L}$, no neurons in $reps(c)$ fire at time $level(c)$.

\begin{proof}
Fix $B \subseteq C_0$, and consider the unique executions of $\mathcal{A}_2$ and $\mathcal L$ on input $B$.
We prove the following statement $Q(t)$, for every $t \geq 0$:

$Q(t)$:  For every concept $c$ with $level(c) = t$, if $rep(c)$ does not fire at time $t$ in $\mathcal{A}_2$, then in $\mathcal{L}$, no neurons in $reps(c)$ fire at time $t$.

We prove $Q(t)$ by induction on $t$, for all concepts $c$.

\emph{Base:}  $t = 0$: 
Suppose that $c$ is a concept with $level(c) = 0$ and $rep(c)$ does not fire at time $0$ in $\mathcal{A}_2$.
Then by definition of presentation in $\mathcal{A}_2$, $c \notin B$.  
Then by definition of presentation in $\mathcal L$, none of the neurons in $reps(c)$ fire at time $0$, as needed.

\emph{Inductive step:}  $t \geq 1$:
Assume that $Q(t-1)$ holds, and consider $c$ with $level(c) = t$.
Suppose that $rep(c)$ does not fire at time $t$ in $\mathcal{A}_2$.
Then $rep(c)$ must not have enough incoming potential from its children to fire.
Since the threshold in $\mathcal{A}_2$ is $r_1 k$, this means that the $reps$ of strictly fewer than $r_1 k$ of $c$'s children fire at time $t-1$.

Then by inductive hypothesis $Q(t-1)$, we have that, in $\mathcal L$, for each child $c'$ of $c$ that does not fire, none of the neurons in $reps(c')$ fire at time $t-1$.
This yields a total incoming potential to each neuron $v \in reps(c)$ of strictly less than $r_1 k m$.  But the firing threshold for each such $v$ is $a r_2 k m (1 - \epsilon)$.
Since we have assumed that $r_1 \leq a r_2 (1 - \epsilon)$, we know that 
$r_1 k m \leq a r_2 k m (1 - \epsilon)$.  
So the incoming potential to each such $v$ is strictly less than its firing threshold $a r_2 k m (1 - \epsilon)$, which implies that $v$ does not fire at time $t$, as needed.
\end{proof}

Now we can show that $\mathcal L$ satisfies the non-firing requirement.

\begin{theorem}
\label{thm: L-non-firing-guarantee-2}
$\mathcal{L}$ satisfies Part 2 of the definition of the recognition problem for networks with multi-neuron representations, Definition~\ref{def: recog-ff-multi}.
\end{theorem}

\begin{proof}
Theorem~\ref{thm: A2-non-firing-guarantee} says that $\mathcal{A}_2$ guarantees $r_1$-non-firing for $\mathcal{C}$.  
Theorem~\ref{thm: L-non-firing-guarantee} says that $\mathcal{L} \leq_{impl2} \mathcal{A}_2$.
Then Theorem~\ref{thm: general-mapping-theorem2} implies that $\mathcal{L}$ satisfies Part 2 of Definition~\ref{def: recog-ff-multi}.
\end{proof}

\subsection{Correctness of $\mathcal L$}

\begin{theorem}
$\mathcal{L}$ $(r_1,r_2)$-recognizes $\mathcal C$.
\end{theorem}

\begin{proof}
By Theorems~\ref{thm: L-firing-guarantee-2} and~\ref{thm: L-non-firing-guarantee-2}.
\end{proof}

Thus, we have shown correctness of the deterministic network $\mathcal{L}$.  The corresponding network of~\cite{lynch2024multineuron} is probabilistic, in that it chooses the failed set $F$ randomly.
It is shown in~\cite{lynch2024multineuron} that, with high probability, that network achieves the restrictions on failures that are assumed here.
Therefore, we can make a claim about the probabilistic version of the network:  that with high probability it $(r_1,r_2)$-recognizes $\mathcal C$.\footnote{Again, we remark that the probabilities here are slightly larger than in~\cite{lynch2024multineuron}, because the "shapes" of the hierarchies are slightly different.} 

\section{Conclusions}
\label{sec: conclusions} 

In this paper, we have demonstrated how detailed, failure-prone Spiking Neural Networks can be proved to work correctly by relating them to more abstract, reliable networks.  
We focus here on the problem of recognizing structured concepts in the presence of partial information, 
The detailed networks use multiple neurons to represent each concept, whereas the abstract networks use only one.  The resulting proofs nicely decompose the reasoning into high-level arguments about the correctness of recognition and lower-level arguments about how networks with multi-neuron representations can emulate networks with single-neuron representations.

Specifically, we consider two detailed networks:  $\mathcal H$, a failure-prone, feed-forward network with full connectivity between layers, and $\mathcal L$, a failure-prone feed-forward network with partial connectivity.
Both $\mathcal H$ and $\mathcal L$ are based on similar networks in~\cite{lynch2024multineuron}.
We show that each of $\mathcal H$ and $\mathcal L$ implements two abstract networks $\mathcal{A}_1$ and $\mathcal{A}_2$, where $\mathcal{A}_1$ expresses firing guarantees and $\mathcal{A}_2$ expresses non-firing guarantees.
We use the two notions of implementation $\leq_{impl1}$ and $\leq_{impl2}$ to relate the detailed networks to $\mathcal{A}_1$ and $\mathcal{A}_2$, respectively.

One difference between our networks and those in~\cite{lynch2024multineuron} is that the networks in~\cite{lynch2024multineuron} are probabilistic, based on random failures.
Here we abstract away from issues of probability, by assuming that the failures are well distributed with respect to various sets of surviving neurons and with respect to edge connections.
Our presentation also differs from that in~\cite{lynch2024multineuron}, in that in~\cite{lynch2024multineuron}, our proofs intertwine reasoning about correctness of recognition with reasoning about multi-neuron representations.
Here we separate these two aspects.

This paper, as well as~\cite{lynch2024multineuron}, were partially inspired by work on the assembly calculus~\cite{PapadimitriouVempala,DBLP:conf/innovations/Legenstein0PV18,DBLP:series/lncs/0001PVL19}.  
That work studies multi-neuron representations of structured concepts in networks with partial connectivity, but does not consider the main issue we emphasize here, namely, how mechanisms using multi-neuron representations might be regarded as implementations of simpler mechanisms using single-neuron representations.

Another difference is that much of the assembly calculus work focuses on learning, not just representation.  
Presumably one could study a correspondence between learning in networks with multi-neuron representations and networks with single-neuron representations, but we have not yet considered this.

%


\paragraph{Future work:}
The approach of this paper basically involves embedding concept structures in Spiking Neural Networks in two ways:  using single-neuron representations, and using multi-neuron representations.
Then we relate the two types of representation using formal implementation relationships, and use these connections to prove properties of the detailed networks based on properties of the abstract networks.

This approach might be applied to study many other neural networks.
For starters, our paper~\cite{lynch2024multineuron} contains a third detailed network for recognizing hierarchical concepts, in addition to the two we adapted for this paper.
That network has lateral edges, in addition to forward edges, and is more closely related to the assembly calculus work than the strictly feed-forward networks studied here.
It would be interesting to extend the approach of this paper to that network,
by abstracting away from its use of probability, introducing appropriate abstract versions, and demonstrating implementation relationships between the detailed network and the abstract networks.

Other networks for recognizing hierarchical concepts might allow some feedback edges in addition to forward edges, as in~\cite{DBLP:conf/sirocco/LynchM23}.  It remains to define versions of those networks with multi-neuron representations, and to try to understand them by relating them formally to networks with single-neuron representations.

Our approach could be useful in studying some brain mechanisms that have been studied previously using the assembly calculus.  The usual assembly calculus approach models the mechanisms in detail, and attempts to understand them directly via simulation and analysis.
But since the mechanisms are complicated, this does not seem easy.  It might be useful to first consider abstract versions of the networks, and ascertain that they work as desired.  
Then one could study the detailed versions by showing that they implement the abstract versions.

In this paper, our concept structures are simple hierarchies.
However, we might also consider more general concept structures, modeled as arbitrary labeled directed graphs.
The abstract networks could involve rather direct embeddings in which each graph vertex is represented by a single node, whereas the detailed networks could include multiple representations for each node, neuron failures, and partial connectivity.

Viewed very generally, our approach involves defining a detailed, realistic version of a network, then trying to understand its behavior by relating it formally to an abstract version, or abstract versions, of the network.
This general approach seems promising as a way of understanding and analyzing complex brain mechanisms.


\bibliography{Multi}

\begin{thebibliography}{10}

\bibitem{DBLP:conf/innovations/Legenstein0PV18}
Robert Legenstein, Wolfgang Maass, Christos~H. Papadimitriou, and Santosh~S. Vempala.
\newblock Long term memory and the densest k-subgraph problem.
\newblock In Anna~R. Karlin, editor, {\em 9th Innovations in Theoretical Computer Science Conference, {ITCS} 2018, January 11-14, 2018, Cambridge, MA, {USA}}, volume~94 of {\em LIPIcs}, pages 57:1--57:15. Schloss Dagstuhl - Leibniz-Zentrum f{\"{u}}r Informatik, 2018.

\bibitem{LM21}
Nancy Lynch and Frederik Mallmann-Trenn.
\newblock Learning hierarchically structured concepts.
\newblock {\em Neural Networks}, 143:798--817, November 2021.

\bibitem{lynch2024multineuron}
Nancy~A. Lynch.
\newblock Multi-neuron representations of hierarchical concepts in spiking neural networks, January 2024.
\newblock arXiv:2401.04628.

\bibitem{DBLP:conf/sirocco/LynchM23}
Nancy~A. Lynch and Frederik Mallmann{-}Trenn.
\newblock Learning hierarchically-structured concepts {II:} overlapping concepts, and networks with feedback.
\newblock In Sergio Rajsbaum, Alkida Balliu, Joshua~J. Daymude, and Dennis Olivetti, editors, {\em Structural Information and Communication Complexity - 30th International Colloquium, {SIROCCO} 2023, Alcal{\'{a}} de Henares, Spain, June 6-9, 2023, Proceedings}, volume 13892 of {\em Lecture Notes in Computer Science}, pages 46--86. Springer, 2023.

\bibitem{DBLP:conf/podc/LynchT87}
Nancy~A. Lynch and Mark~R. Tuttle.
\newblock Hierarchical correctness proofs for distributed algorithms.
\newblock In Fred~B. Schneider, editor, {\em Proceedings of the Sixth Annual {ACM} Symposium on Principles of Distributed Computing, Vancouver, British Columbia, Canada, August 10-12, 1987}, pages 137--151. {ACM}, 1987.

\bibitem{DBLP:series/lncs/0001PVL19}
Wolfgang Maass, Christos~H. Papadimitriou, Santosh~S. Vempala, and Robert Legenstein.
\newblock Brain computation: {A} computer science perspective.
\newblock In Bernhard Steffen and Gerhard~J. Woeginger, editors, {\em Computing and Software Science - State of the Art and Perspectives}, volume 10000 of {\em Lecture Notes in Computer Science}, pages 184--199. Springer, 2019.

\bibitem{PapadimitriouVempala}
Christos~H. Papadimitriou and Santosh~S. Vempala.
\newblock Random projection in the brain and computation with assemblies of neurons.
\newblock In {\em 10th Innovations in Theoretical Computer Science Conference, ITCS 2019}, pages 57:1--57:19, San Diego, California, January 2019.

\bibitem{PVMM20}
Christos~H. Papadimitriou, Santosh~S. Vempala, Daniel Mitropolsky, and Wolfgang Maass.
\newblock Brain computation by assemblies of neurons.
\newblock {\em Proceedings of the National Academy of Sciences}, 2020.

\bibitem{DBLP:journals/njc/SegalaL95}
Roberto Segala and Nancy~A. Lynch.
\newblock Probabilistic simulations for probabilistic processes.
\newblock {\em Nord. J. Comput.}, 2(2):250--273, 1995.

\bibitem{Valiant}
Leslie~G. Valiant.
\newblock {\em Circuits of the Mind}.
\newblock Oxford University Press, 2000.

\end{thebibliography}

\end{document}